\documentclass[a4paper,UKenglish]{lipics-v2016}

\usepackage{microtype}

\usepackage{stmaryrd,proof,amssymb,amsmath}
\usepackage[matrix,arrow,curve]{xy}

\newcommand{\Par}{\mathop{\bindnasrepma}}
\newcommand{\Tensor}{\otimes}

\newcommand{\LL}{\mathbf{L^{\!\!*}}}
\newcommand{\LC}{\mathbf{L}}
\newcommand{\Lb}{\mathbf{Lb^{\!*}}}

\newcommand{\SL}{\mathop{/}}
\newcommand{\BS}{\mathop{\backslash}}

\newcommand{\Ec}{\mathcal{E}}
\newcommand{\Ac}{\mathcal{A}}
\newcommand{\Gf}{\mathfrak{G}}
\newcommand{\Af}{\mathfrak{A}}
\newcommand{\Sc}{\mathcal{S}}

\newcommand{\poly}{\mathrm{poly}}

\newcommand{\metapar}{\diamond}

\newcommand{\PMod}{\langle\rangle}
\newcommand{\NMod }{[]^{-1}}

\newcommand{\CN}{\mbox{$\mathit{CN}$}}

\newcommand{\nrm}[1]{|\!|#1|\!|}

\newcommand{\ord}{\mathrm{ord}}
\newcommand{\prd}{\mathrm{prod}}
\newcommand{\br}{\mathrm{b}}

\EventEditors{Dale Miller}
\EventNoEds{1}
\EventLongTitle{2nd International Conference on Formal Structures for Computation and Deduction (FSCD 2017)}
\EventShortTitle{FSCD 2017}
\EventAcronym{FSCD}
\EventYear{2017}
\EventDate{September 3--9, 2017}
\EventLocation{Oxford, UK}
\EventLogo{}
\SeriesVolume{84}
\ArticleNo{22}

\title{A Polynomial-Time Algorithm for the Lambek Calculus with Brackets of Bounded Order}
\titlerunning{A Poly-Time Algorithm for the Lambek Calculus with Brackets of Bounded Order}
\author[1]{Max Kanovich}
\author[2,1]{Stepan Kuznetsov}
\author[3]{Glyn Morrill}
\author[4,1]{Andre Scedrov}
\affil[1]{National Research University Higher School of Economics, Moscow, Russia\\ \texttt{mkanovich@hse.ru}}
\affil[2]{Steklov Mathematical Institute of RAS, Moscow, Russia\\ \texttt{sk@mi.ras.ru}}
\affil[3]{Universitat Politècnica de Catalunya, Barcelona, Spain\\ \texttt{morrill@cs.upc.edu}}
\affil[4]{University of Pennsylvania, Philadelphia, USA\\ \texttt{scedrov@math.upenn.edu}}
\authorrunning{M. Kanovich, S. Kuznetsov, G. Morrill, A. Scedrov}
\Copyright{Max Kanovich, Stepan Kuznetsov, Glyn Morrill, Andre Scedrov}


\begin{document}

\maketitle

\begin{abstract}
Lambek calculus is a logical foundation of categorial grammar, a linguistic paradigm of grammar as logic and parsing as deduction. Pentus (2010) gave a polynomial-time algorithm for determining provability of bounded depth formulas in $\LL$, the Lambek calculus with empty antecedents allowed. Pentus' algorithm is based on tabularisation of proof nets. Lambek calculus with brackets is a conservative extension of Lambek calculus with bracket modalities, suitable for the modeling of syntactical domains. In this paper we give an algorithm for provability in $\Lb$, the Lambek calculus with brackets allowing empty antecedents. Our algorithm runs in polynomial time when both the formula depth and the bracket nesting depth are bounded. It
 combines a Pentus-style tabularisation of proof nets with an automata-theoretic treatment of bracketing.
\end{abstract}
 
\section{Introduction}\label{S:intro}

The calculus $\LC$ of Lambek~\cite{Lambek1958} is a logic of strings.
It is retrospectively recognisable as the multiplicative fragment of non-commutative
intuitionistic linear logic without empty antecedents; the calculus $\LL$ is like $\LC$
except that it admits empty antecedents.
The Lambek calculus is the foundation of categorial grammar,
a linguistic paradigm of grammar as logic and parsing as deduction;
see for instance Buszkowski~\cite{Buszkowski}, Carpenter~\cite{carpenter:96}, J\"{a}ger~\cite{Jaeger2005}, 
Morrill~\cite{Morrill2011}, 
Moot and Retor\'{e}~\cite{MootRetore2012}. 
For example, 
the sentence ``{\sl John knows Mary likes Bill}'' can be analysed as grammatical because 
$N, (N \BS S) \SL S, N, (N \BS S) \SL N, N \to S$  is a
theorem of Lambek calculus. Here $N$ stands for {\em noun phrase,} $S$ stands for 
{\em sentence,} and syntactic categories for other words are built from these two
primitive ones using division operations. For example, $(N \BS S) \SL N$ takes
noun phrases on both sides and yields a sentence, thus being the category of
{\em transitive verb.}

Categorial grammar, that started from works of Ajdukiewicz~\cite{ajdukiewicz} and
Bar-Hillel~\cite{bar-hillel:quasi}, aspires to practice linguistics to the standards of mathematical logic;
for example, Lambek~\cite{Lambek1958} proves cut-elimination, that yields
 the subformula property, decidability, the finite reading property,
and the focalisation property.
In a remarkable series of works Mati Pentus has proved the main metatheoretical results
for Lambek calculus:
equivalence to
context free grammars~\cite{pentus:lccf};
completeness w.r.t.\ language models~\cite{Pentus95APAL}\cite{Pentus1998};
NP-completeness~\cite{pentus:npcomplete};
a polynomial-time algorithm for checking provability of formulae of bounded order in $\LL$~\cite{Pentus2010}. 
The Lambek calculus with only one division operation (and without product) is decidable in polynomial
time (Savateev~\cite{savateev:polynomial}).

The Lambek calculus with  brackets $\mathbf{Lb}$ (Morrill 1992~\cite{Morrill1992}; Moortgat 1995~\cite{Moortgat1995})
is a logic of bracketed strings which is a conservative extension of Lambek calculus with bracket
modalities the rules for which are conditioned on metasyntactic brackets. In this paper we consider a variant
of $\mathbf{Lb}$ that allows empty antecedents, denoted by $\Lb$.

The syntax of $\Lb$ is more involved than the syntax of the original Lambek calculus. In $\LC$, the antecedent
(left-hand side) of a sequent is just a linearly ordered sequence of formulae. In $\Lb$, it is a structure called
{\em configuration,} or {\em meta-formula.} Meta-formulae are built from formulae, or {\em types,} as they are called
in categorial grammar, using two metasyntactic constructors: comma and brackets. The succedent (right-hand side) of a sequent
is one type. Types, in turn, are built from variables, or {\em primitive types,} $p_1$, $p_2$, \ldots, using the three
binary connectives of Lambek, $\BS$, $\SL$, and $\cdot$, and two unary ones, $\PMod$ and $\NMod$, that operate brackets.
Axioms of $\Lb$ are $p_i \to p_i$, and the rules are as follows:

$$
\infer[(\BS\to)]{\Delta ( \Pi, A \BS B ) \to C}
{\Pi \to A & \Delta ( B ) \to C}
\qquad
\infer[(\to\BS)]{\Pi \to A \BS B}{A, \Pi \to B}
\qquad
\infer[(\cdot\to)]{\Gamma ( A \cdot B ) \to C}
{\Gamma ( A, B ) \to C}
$$
$$
\infer[(\SL\to)]{\Delta ( B \SL A, \Pi ) \to C}
{\Pi \to A & \Delta ( B ) \to C}
\qquad
\infer[(\to\SL)]{\Pi \to B \SL A}{\Pi, A \to B}
\qquad
\infer[(\to\cdot)]{\Gamma, \Delta \to A \cdot B}{\Gamma \to A & \Delta \to B}
$$
$$
\infer[(\PMod\to)]{\Delta ( \PMod A ) \to C}
{\Delta ( [ A ] ) \to C}
\quad
\infer[(\to\PMod)]{[\Pi] \to \PMod A}{\Pi \to A}
\quad
\infer[(\NMod \to)]{\Delta ( [ \NMod  A ] ) \to C}
{\Delta ( A ) \to C}
\quad
\infer[(\to\NMod )]{\Pi \to \NMod  A}{[\Pi] \to A}
$$

Cut-elimination is proved in Moortgat~\cite{Moortgat1995}.
The Lambek calculus with brackets permits the characterisation of syntactic domains in addition
to word order. 
By way of linguistic example, consider how a relative pronoun type assignment
$(\CN\BS\CN) \SL (S \SL N)$ (here $CN$ is {\em one} primitive type,
corresponding to {\em common noun:} e.g., ``{\sl book}'', as opposed to noun phrase
``{\sl the book}'') allows unbounded relativisation by associative assembly
of the body of relative clauses:

\begin{tabular}[t]{ll}
(a) & {\sl man who Mary likes}\\
(b) & {\sl man who John knows Mary likes}\\
(c) & {\sl man who Mary knows John knows Mary likes} \qquad \dots
\end{tabular}

\noindent
Thus, (b) is generated because the following is a theorem in the pure Lambek calculus:
$$
\CN, (\CN \BS \CN) \SL (S \SL N), N, (N\BS S) \SL S, N, (N\BS S) \SL N \to \CN
$$
Consider also, however, the following example:
{\sl *book which John laughed without reading,}
where * indicates that this example is not grammatical.
In the original Lambek calculus this ungrammatical example is generated, but
in Lambek calculus with brackets its ungrammaticality can be characterised by assigning the adverbial preposition
a type $\NMod((N \BS S) \BS (N \BS S)) \SL (N \BS S)$ blocking this phrase because
the following is not a theorem in Lambek calculus with brackets:
$$
\CN, (\CN \BS \CN) \SL (S \SL N), N, N \BS S, [\NMod((N\BS S)\BS(N\BS S))\SL(N\BS S), (N\BS S)\SL N]\to\CN$$
where the $\NMod$ engenders brackets which block the associative assembly of the body
of the relative clause.
Another example of islands is provided by the ``and'' (``or'') construction:
{\sl *girl whom John loves Mary and Peter loves.} 

J\"ager~\cite{Jaeger2003} claims to prove the context free equivalence of {\bf Lb} grammar on the basis of
a translation from {\bf Lb} to {\bf L} due to Michael Moortgat's student Koen Versmissen~\cite{Versmissen1996}.
However, contrary to Versmissen the translation is not an embedding
translation (Fadda and Morrill~\cite{FaddaMorrill2005}, p.~124). 
We present the counter-example in the end of Section~\ref{S:proofnets}.
Consequently the result
of J\"ager is in doubt: the context free equivalence theorem might be correct,
but the proof of J\"ager, resting on the Versmissen translation, is not correct.

Pentus~\cite{Pentus2010}, following on Aarts~\cite{Aarts}, presents
an algorithm for provability in $\LL$ based on tabularisation (memoisation) of proof nets.
This algorithm runs in polynomial time, if the order of the sequent is bounded.
An algorithm of the same kind was also developed by Fowler~\cite{Fowler2008}\cite{Fowler2009}
for the Lambek calculus without product.
For the unbounded case, the derivability problem for $\LL$ is in the NP class and is NP-complete~\cite{pentus:npcomplete}.
Also, non-associative Lambek calculus~\cite{Lambek1961} can be embedded into the Lambek calculus with brackets
(Kurtonina~\cite{kurtonina:phd}). De Groote~\cite{deGrooteNL}, following on Aarts and Trautwein~\cite{AartsTrautwein},
showed polynomial-time decidability of the non-associative Lambek calculus. The Lambek calculus with bracket modalities, $\Lb$, includes
as subsystems both non-associative and associative Lambek calculi.

In this paper we provide a Pentus-style algorithm
for $\Lb$ provability using (1) the proof nets for Lambek calculus with brackets of 
Fadda and Morrill~\cite{FaddaMorrill2005} which are based on a correction of the Versmissen
translation, and (2) an automata-theoretic argument. Again, for the unbounded case,
$\Lb$ is NP-hard (since it contains $\LL$ as a conservative fragment), and also belongs
to the NP class, since the size of a cut-free derivation in $\Lb$ is linearly bounded by the
size of the goal sequent.

The rest of this paper is organised as follows.
In Section~\ref{S:main} we define complexity parameters and formulate the main result.
Section~\ref{S:proofnets} contains the formulation and proof of a graph-theoretic provability 
criterion for $\Lb$, known as proof nets. 
In Section~\ref{S:params} we introduce some more convenient complexity parameters and show
their polynomial equivalence to the old ones.
Section~\ref{S:algo} is the central one, containing the
description of our algorithm. 
In order to make this paper self-contained, in Section~\ref{S:Pentus} we give a detailed explanation of Pentus' construction~\cite{Pentus2010},
since it is crucial for our algorithm to work. 
Finally, in Section~\ref{S:future} we discuss directions of future research in this field.

\section{The Main Result}\label{S:main}

For a sequent $\Gamma \to C$ we consider the following three {\em complexity parameters.} The first one is the {\em size}
of the sequent, $\nrm{\Gamma \to C}$, counted as the total number of variables and logical symbols in it, including brackets.

\begin{definition}
The size of a formula, meta-formula, or sequent in $\Lb$ is defined recursively as follows:
$\nrm{p_i} = 0$; $\nrm{A \cdot B} = \nrm{A \BS B} = \nrm{B \SL A} = \nrm{A} + \nrm{B} + 1$;
$\nrm{\PMod A} = \nrm{\NMod A} = \nrm{A} + 1$; $\nrm{\Lambda} = 0$;
$\nrm{\Gamma, \Delta} = \nrm{\Gamma} + \nrm{\Delta}$; $\nrm{[\Gamma]} = \nrm{\Gamma} + 2$;
$\nrm{\Gamma \to C} = \nrm{\Gamma} + \nrm{C}$.
\end{definition}
The second parameter is the {\em order.}

\begin{definition}
For any formula $A$ let $\prd(A)$ be $1$ if $A$ is of the form $A_1 \cdot A_2$ or $\PMod A_1$, and
$0$ if not. The order of a formula, meta-formula, or sequent in $\Lb$ is defined recursively:
$\ord(p_i) = 0$; $\ord(A \cdot B) = \max \{ \ord(A), \ord(B) \}$;
$\ord(A \BS B) = \ord(B \SL A) = \max \{ \ord(A) + 1, \ord(B) + \prd(B) \}$;
$\ord(\PMod A) = \ord(A)$;
$\ord(\NMod A) = \max \{ \ord(A) + \prd(A), 1 \}$;
$\ord(\Lambda) = 0$;
$\ord(\Gamma, \Delta) = \max \{ \ord(\Gamma), \ord(\Delta) \}$;
$\ord([\Gamma]) = \ord(\Gamma)$;
$\ord(\Gamma \to C) = \max \{ \ord(\Gamma) + 1, \ord(C) + \prd(C) \}$.
\end{definition}
For sequents without $\cdot$ and $\NMod$, this definition is quite intuitive: the order is the nesting depth of
implications ($\BS$, $\SL$, and finally $\to$) and $\PMod$ modalities. With $\cdot$, we also count {\em alternations}
between divisions and multiplications: for example, in $p_1 \BS (p_2 \cdot (p_3 \BS (p_4 \cdot \ldots p_k) \ldots))$
implications are not nested, but the order grows linearly. On the other hand, the order is always bounded by
a simpler complexity parameter, the maximal height of the syntactic tree. Also, linguistic applications make use
of syntactic types of small, constantly bounded order.

The third parameter is the {\em bracket nesting depth}.
\begin{definition}
The bracket nesting depth of a formula, meta-formula, or a sequent in $\Lb$ is defined recursively as follows:
$\br(p_i) = 0$;
$\br(A \SL B) = \br(B \BS A) = \br(A \cdot B) = \max \{ \br(A), \br(B) \}$;
$\br(\PMod A) = \br(\NMod  A) = \br(A) + 1$;
$\br(\Lambda) = 0$;
$\br(\Gamma, \Delta) = \max \{ \br(\Gamma), \br(\Delta) \}$;
$\br([\Gamma]) = \br(\Gamma) + 1$;
$\br(\Gamma \to C) = \max \{ \br(\Gamma), \br(C) \}$.
\end{definition}

By $\poly(x_1, x_2, \ldots)$ we denote a value that is bounded by a polynomial of $x_1$, $x_2$, \ldots

\begin{theorem}\label{Th:main}
There exists an algorithm that decides whether a sequent $\Gamma \to C$ is derivable in $\Lb$ in 
$\poly(N, 2^R, N^B)$ time, where $N = \nrm{\Gamma \to C}$, $R = \ord(\Gamma \to C)$, 
$B = \br(\Gamma \to C)$.
\end{theorem}
If the depth parameters, $R$ and $B$, are fixed, the working time of the algorithm is polynomial
 w.r.t.~$N$. However, the dependence on the depth parameters is exponential.

\section{Proof Nets}\label{S:proofnets}
In this section we formulate and prove a graph-theoretic criterion for derivability in $\Lb$.
A sequent is derivable if and only if there exists a {\em proof net,} that is, a graph satisfying certain
{\em correctness conditions.} 

For each variable $p_i$ we introduce two {\em literals,} $p_i$ and $\bar{p}_i$, and also
four literals, $[$, $]$, $\bar{[}$, and $\bar{]}$ for brackets.  
Next we define two translations (positive, $A^+$, and negative, $A^-$) of $\Lb$-formulae
into expressions built from literals using two connectives, $\Par$ and $\Tensor$.

\begin{definition}\label{Df:translation}
\begin{align*}
& p_i^+ = p_i,		&& p_i^- = \bar{p_i},\\
& (A \cdot B)^+ = A^+ \Tensor B^+,	&& (A \cdot B)^- = B^- \Par A^-, \\
& (A \BS B)^+ = A^- \Par B^+,		&& (A \BS B)^- = B^- \Tensor A^+,\\
& (B \SL A)^+ = B^+ \Par A^-,		&& (B \SL A)^- = A^+ \Tensor B^-,\\
& (\PMod A)^+ = {]} \Tensor A^+ \Tensor {[}, && (\PMod A)^- = \bar{[} \Par A^- \Par \bar{]},\\
& (\NMod  A)^+ = \bar{]} \Par A^+ \Par \bar{[}; && (\NMod A)^- = {[} \Tensor A^- \Tensor {]}.
\end{align*}
\end{definition}

For meta-formulae, we need only the negative translation. In this translation we use
an extra connective, $\metapar$, which serves as a metasyntactic version of $\Par$
(just as the comma is a metasyntactic product in the sequent calculus for $\Lb$).
\begin{definition}
$(\Gamma, \Delta)^- = \Delta^- \metapar \Gamma^-$; \quad $[\Gamma]^- = \bar{[} \metapar \Gamma^- \metapar \bar{]}$.
\end{definition}
Finally, a sequent $\Gamma \to C$ is translated as ${}\metapar \Gamma^- \metapar C^+$ (or as ${}\metapar C^+$ if $\Gamma$ is empty).

Essentially, this is an extension of Pentus' translation of $\LL$ into cyclic multiplicative linear logic ({\bf CMLL})~\cite{Pentus1998}\cite{Pentus2010}.
In this paper, for the sake of simpicity, we don't introduce an intermediate calculus that extends {\bf CMLL} with brackets,
and we formulate the proof net criterion directly for $\Lb$.

Denote the set of all literal and connective {\em occurrences} in this translation 
by $\Omega_{\Gamma \to C}$. These occurrences are linearly ordered; connectives and literals alternate. The total number
of occurrences is $2n$. Denote
the literal occurrences (in their order) by $\ell_1, \dots, \ell_n$ and the connective occurrences by $c_1, \dots, c_n$.

\begin{definition}\label{Df:dom}
The {\em dominance} relation on the occurrences
of $\Par$ and $\Tensor$, denoted by $\prec$, is defined as follows: for a subexpression of the form $A \Par B$ or $A \Tensor B$ if the occurrence
of the central connective is $c_i$, then for any  $c_j$ inside $A$ or $B$ we declare $c_j \prec c_i$.
\end{definition}

We assume that $\Par$'s that come from translations of bracket modalities associate to the left and $\Tensor$'s associate to
the right (this choice is arbitrary). Thus, in a pair of such $\Par$'s the right one dominates the left one in the 
syntactic tree, and the left one dominates the principal connective of $A$ (if $A$ has one, {\em i.e.,} it is not a literal); symmetrically for $\Tensor$.

The other two relations are the {\em sisterhood} relations on bracket literals and connectives, $\Sc_b$ and $\Sc_c$ respectively.
Both relations are symmetric. 

\begin{definition}
The bracket sisterhood relation, $\Sc_b$, connects pairs of occurrences of $[$ and $]$ or $\bar{[}$ and $\bar{]}$ that
come from the same $\PMod A$, $\NMod A$, or $[\Gamma]$. The connective sisterhood relation, $\Sc_c$, connects pairs of occurrences of $\Tensor$,
$\Par$, or $\metapar$ that come from the same $\PMod A$, $\NMod A$, or $[\Gamma]$. Occurrences connected by one of the
sisterhood relations will be called {\em sister} occurrences.
\end{definition}

\begin{definition}
A {\em proof structure} $\Ec$ is a symmetric relation on the
set of literal occurrences ($\{ \ell_1, \dots, \ell_n \}$) such that each occurrence is connected by $\Ec$ to exactly one occurrence,
and each occurrence of a literal $q$, where $q$ is a variable, $[$, or $]$, is connected to an occurrence of $\bar{q}$.
\end{definition}

\begin{definition}
A proof structure $\Ec$ is {\em planar,} if{f} its edges can be drawn in a semiplane without intersection while the literal
occurrences are located on the border of this semiplane in their order ($\ell_1, \dots, \ell_n$).
\end{definition}

Edges of a planar proof structure divide the upper semiplane into {\em regions.} The number of regions is
$\frac{n}{2} + 1$ (the outermost, infinite region also counts).

\begin{definition}\label{Df:proofnet}
A planar proof structure $\Ec$ is a {\em proof net,} if{f} it satisfies two conditions.
\begin{enumerate}
\item On the border of each region there should be exactly one occurrence of $\Par$ or $\metapar$. 
\item Define an oriented graph $\Ac$  that connects each occurrence of $\Tensor$ to the unique occurrence of $\Par$ or $\metapar$ located
in the same region. The graph $\Ac \cup {\prec}$ should be acyclic.
\end{enumerate}
\end{definition}

By definition, edges of $\Ec$ and $\Ac$ in a proof net do not intersect, in other words, the graph $\Ec \cup \Ac$ is also planar.
We can also consider proof structures and proof nets on expressions that are not translations of $\Lb$ sequents. For example, proof nets allow
{\em cyclic permutations:} a proof net for ${}\metapar \gamma_1 \metapar \gamma_2$ can be transformed into a proof net
for ${}\metapar \gamma_2 \metapar \gamma_1$ (the $\prec$ relation in $\gamma_1$ and $\gamma_2$ is preserved).

\begin{definition}
A proof structure $\Ec$ {\em respects sisterhood,} if{f} the following condition holds: 
if $\langle \ell_i, \ell_{i'} \rangle \in \Ec$, $\langle \ell_i, \ell_j \rangle \in \Sc_b$, and
$\langle \ell_{i'}, \ell_{j'} \rangle \in \Sc_b$, then $\langle \ell_j, \ell_{j'} \rangle \in \Ec$
(i.e., sister brackets are connected to sister brackets).
\end{definition}

Before continuing, we consider an {\bf example} of a proof net for a sequent with linguistic meaning. 
According to~\cite{Morrill2011}, the sentence {\sl ``Mary danced before singing''} gets the following bracketing:
{\sl ``}[{\sl Mary}] {\sl danced }[{\sl before singing}]{\sl ''} and the following type assignment:
$$
[ N ], \PMod N \BS S, [ \NMod(( \PMod N \BS S) \BS (\PMod N \BS S)) \SL (\PMod N \BS S), \PMod N \BS S ] \to S
$$
This sequent is derivable in~$\Lb$, and we justify it by presenting a proof net. This is achieved
by translating the sequent into a string of literals according to Definition~\ref{Df:translation} and drawing
an appropriate $\Ec$ graph that satisfies all the conditions for being a proof net (Definition~\ref{Df:proofnet}):
$$\small
\xymatrix @-10mm{
\metapar & \bar{[} 		\ar@{-} '+/u12mm/'[rrrrrrrrrrrrrrrrrr]+/u12mm/ [rrrrrrrrrrrrrrrrrr]
& \metapar & \bar{S} 	\ar@{-} '+/u10mm/'[rrrrrrrrrrrrrr]+/u10mm/ [rrrrrrrrrrrrrr]
& \Tensor & ]			\ar@{-} '+/u8mm/'[rrrrrrrrrr]+/u8mm/ [rrrrrrrrrr]
& \Tensor & N			\ar@{-} '+/u6mm/'[rrrrrr]+/u6mm/ [rrrrrr]
& \Tensor & [ 			\ar@{-} '+/u4mm/'[rr]+/u4mm/ [rr]
& \metapar & \bar{[} 
& \Par & \bar{N} 
& \Par & \bar{]} 
& \Par & S 
& \Tensor & [ 
& \Tensor & \bar{S} 		\ar@{-} '+/u20mm/'[rrrrrrrrrrrrrrrrrrrrrrrrrrrrrrrrrr]+/u20mm/ [rrrrrrrrrrrrrrrrrrrrrrrrrrrrrrrrrr]
& \Tensor & ]			\ar@{-} '+/u18mm/'[rrrrrrrrrrrrrrrrrrrrrrrrrrrrrr]+/u18mm/ [rrrrrrrrrrrrrrrrrrrrrrrrrrrrrr]
& \Tensor & N			\ar@{-} '+/u16mm/'[rrrrrrrrrrrrrrrrrrrrrrrrrr]+/u16mm/ [rrrrrrrrrrrrrrrrrrrrrrrrrr]
& \Tensor & [			\ar@{-} '+/u14mm/'[rrrrrrrrrrrrrrrrrrrrrr]+/u14mm/ [rrrrrrrrrrrrrrrrrrrrrr]
& \Tensor & \bar{[}		\ar@{-} '+/u12mm/'[rrrrrrrrrrrrrrrrrr]+/u12mm/ [rrrrrrrrrrrrrrrrrr]
& \Par & \bar{N}			\ar@{-} '+/u10mm/'[rrrrrrrrrrrrrr]+/u10mm/ [rrrrrrrrrrrrrr]
& \Par & \bar{]}			\ar@{-} '+/u8mm/'[rrrrrrrrrr]+/u8mm/ [rrrrrrrrrr]
& \Par & S				\ar@{-} '+/u6mm/'[rrrrrr]+/u6mm/ [rrrrrr]
& \Tensor & ]			\ar@{-} '+/u4mm/'[rr]+/u4mm/ [rr]
& \metapar & \bar{]}
& \metapar & \bar{S}
& \Tensor & ]
& \Tensor & N
& \Tensor & [
& \metapar & \bar{[}
& \metapar & \bar{N}
& \metapar & \bar{]}
& \metapar & S
}
$$
Informally speaking, two literals are connected by $\Ec$ if they come from the same axiom
or bracket rule instance. The next figure depicts the $\Ac$ relation obtained from $\Ec$
by Definition~\ref{Df:proofnet} and the dominance relation ${\prec}$ (Definition~\ref{Df:dom}), showing
that $\Ac \cup {\prec}$ is acyclic:
$$\small
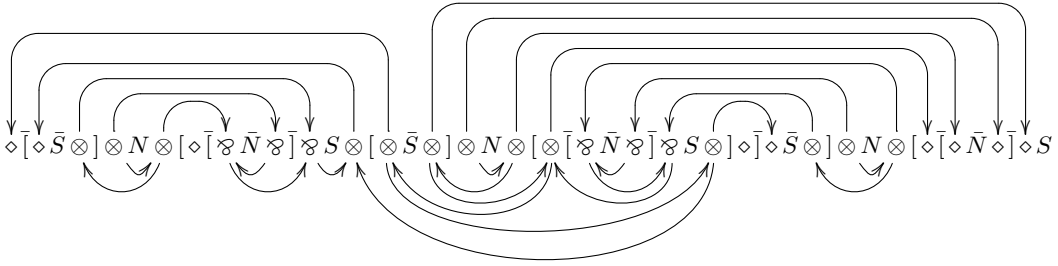
\xymatrix @-10mm{
\metapar & \bar{[} 		
& \metapar & \bar{S} 	
& \Tensor \ar `u_r[rrrrrrrr]+/u9mm/`r_d[rrrrrrrrrrrr] [rrrrrrrrrrrr] & ]		
& \Tensor  \ar `u_r[rrrrrr]+/u7mm/`r_d[rrrrrrrr] [rrrrrrrr]   \ar@{>}@`{-/l4mm/+/d5mm/,p+/l4mm/+/d5mm/} [rr] 
& N			
& \Tensor \ar `u_r[]+/u6mm/`r_d[rrrr] [rrrr] 					\ar@{>}@`{+/l1mm/+/d8mm/,p-/l1mm/+/d8mm/} [llll] 
& [ 			
& \metapar & \bar{[} 
& \Par \ar@{>}@`{-/l1mm/+/d8mm/,p+/l1mm/+/d8mm/} [rrrr]  & \bar{N} 
& \Par \ar@{>}@`{+/l3mm/+/d5mm/,p-/l3mm/+/d5mm/} [ll] & \bar{]} 
& \Par  \ar@{>}@`{-/l2mm/+/d5mm/,p+/l2mm/+/d5mm/} [rr]  & S 
& \Tensor \ar `u^l[llllllllllllllll]+/u11mm/`l^d[llllllllllllllll] [llllllllllllllll] & [ 
& \Tensor \ar `u^l[llllllllllllllllllll]+/u15mm/`l^d[llllllllllllllllllll] [llllllllllllllllllll] 
\ar@{>}@`{-/r4mm/+/d15mm/,p-/r2mm/+/d15mm/} [rrrrrrrrrrrrrrrr]
 & \bar{S} 		
& \Tensor \ar `u_r[rrrrrrrrrrrrrrrrrrrrrrrrrrrr]+/u19mm/`r_d[rrrrrrrrrrrrrrrrrrrrrrrrrrrrrrrr] [rrrrrrrrrrrrrrrrrrrrrrrrrrrrrrrr] 
\ar@{>}@`{+/l3mm/+/d10mm/,p+/l1mm/+/d10mm/} [rrrrrr] & ]		
& \Tensor \ar `u_r[rrrrrrrrrrrrrrrrrrrrrrrr]+/u17mm/`r_d[rrrrrrrrrrrrrrrrrrrrrrrrrrrr] [rrrrrrrrrrrrrrrrrrrrrrrrrrrr] 
\ar@{>}@`{-/l4mm/+/d5mm/,p+/l4mm/+/d5mm/} [rr]  & N		
& \Tensor \ar `u_r[rrrrrrrrrrrrrrrrrrrr]+/u15mm/`r_d[rrrrrrrrrrrrrrrrrrrrrrrr] [rrrrrrrrrrrrrrrrrrrrrrrr] \ar@{>}@`{+/l1mm/+/d8mm/,p-/l1mm/+/d8mm/} [llll] 
& [		
& \Tensor \ar `u_r[rrrrrrrrrrrrrrrr]+/u13mm/`r_d[rrrrrrrrrrrrrrrrrrrr] [rrrrrrrrrrrrrrrrrrrr]
\ar@{>}@`{+/r2mm/+/d12mm/,p+/r2mm/+/d12mm/} [llllllll]
 & \bar{[}		
& \Par \ar@{>}@`{-/l1mm/+/d8mm/,p+/l1mm/+/d8mm/} [rrrr] & \bar{N}			
& \Par \ar@{>}@`{+/l3mm/+/d5mm/,p-/l3mm/+/d5mm/} [ll] & \bar{]}			
& \Par \ar@{>}@`{-/l3mm/+/d10mm/,p-/l1mm/+/d10mm/} [llllll] & S				
& \Tensor \ar `u_r[]+/u6mm/`r_d[rrrr] [rrrr] 
\ar@{>}@`{+/r2mm/+/d20mm/,p+/r2mm/+/d20mm/} [llllllllllllllllll]
& ]			
& \metapar & \bar{]}
& \metapar & \bar{S}
& \Tensor  \ar `u^l[lllllll]+/u7mm/`l^d[llllllll] [llllllll] & ]
& \Tensor  \ar `u^l[lllllllllll]+/u9mm/`l^d[llllllllllll] [llllllllllll]  \ar@{>}@`{-/l4mm/+/d5mm/,p+/l4mm/+/d5mm/} [rr]  & N
& \Tensor  \ar `u^l[llllllllllllllll]+/u11mm/`l^d[llllllllllllllll] [llllllllllllllll] \ar@{>}@`{+/l1mm/+/d8mm/,p-/l1mm/+/d8mm/} [llll]  & [
& \metapar & \bar{[}
& \metapar & \bar{N}
& \metapar & \bar{]}
& \metapar & S
}
$$

\vspace*{3.5em}\noindent
Here $\Ac$ and ${\prec}$ are drawn above and below the string respectively.

Finally, as one can see from the figure below, the proof net from our example respects sisterhood:
$$
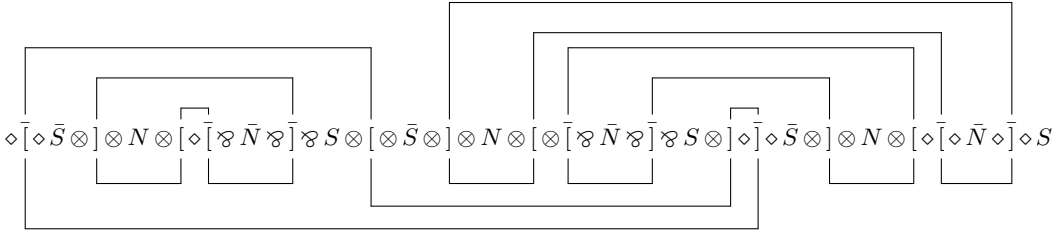
\small
\xymatrix @-10mm{
\metapar & \bar{[} 		\ar@{-} '+/u12mm/'[rrrrrrrrrrrrrrrrrr]+/u12mm/ [rrrrrrrrrrrrrrrrrr]
		\ar@{-} '+/d12mm/'[rrrrrrrrrrrrrrrrrrrrrrrrrrrrrrrrrrrrrr]+/d12mm/ [rrrrrrrrrrrrrrrrrrrrrrrrrrrrrrrrrrrrrr]
& \metapar & \bar{S} 	
& \Tensor & ]			\ar@{-} '+/u8mm/'[rrrrrrrrrr]+/u8mm/ [rrrrrrrrrr]					\ar@{-} '+/d6mm/'[rrrr]+/d6mm/ [rrrr]
& \Tensor & N			
& \Tensor & [ 			\ar@{-} '+/u4mm/'[rr]+/u4mm/ [rr]
& \metapar & \bar{[} 						\ar@{-} '+/d6mm/'[rrrr]+/d6mm/ [rrrr]
& \Par & \bar{N} 
& \Par & \bar{]} 
& \Par & S 
& \Tensor & [ 			
				\ar@{-} '+/d9mm/'[rrrrrrrrrrrrrrrrrr]+/d9mm/ [rrrrrrrrrrrrrrrrrr]
& \Tensor & \bar{S} 		
& \Tensor & ]			\ar@{-} '+/u18mm/'[rrrrrrrrrrrrrrrrrrrrrrrrrrrrrr]+/u18mm/ [rrrrrrrrrrrrrrrrrrrrrrrrrrrrrr]	\ar@{-} '+/d6mm/'[rrrr]+/d6mm/ [rrrr]
& \Tensor & N			
& \Tensor & [			\ar@{-} '+/u14mm/'[rrrrrrrrrrrrrrrrrrrrrr]+/u14mm/ [rrrrrrrrrrrrrrrrrrrrrr]
& \Tensor & \bar{[}		\ar@{-} '+/u12mm/'[rrrrrrrrrrrrrrrrrr]+/u12mm/ [rrrrrrrrrrrrrrrrrr]		\ar@{-} '+/d6mm/'[rrrr]+/d6mm/ [rrrr]
& \Par & \bar{N}			
& \Par & \bar{]}			\ar@{-} '+/u8mm/'[rrrrrrrrrr]+/u8mm/ [rrrrrrrrrr]
& \Par & S				
& \Tensor & ]			\ar@{-} '+/u4mm/'[rr]+/u4mm/ [rr]
& \metapar & \bar{]}
& \metapar & \bar{S}
& \Tensor & ]				\ar@{-} '+/d6mm/'[rrrr]+/d6mm/ [rrrr]
& \Tensor & N
& \Tensor & [
& \metapar & \bar{[}			\ar@{-} '+/d6mm/'[rrrr]+/d6mm/ [rrrr]
& \metapar & \bar{N}
& \metapar & \bar{]}
& \metapar & S
}
$$
Here $\Ec$ (we keep only edges connected to brackets) and $\Sc_b$ are drawn above and below the string respectively.
Graphically the sisterhood condition means that these edges form 4-cycles.
 In the end of this
section we show that this sisterhood condition is essential: if it fails, the sequent could be
not derivable.

Now we are going to prove that a sequent is derivable in $\Lb$ if and only if there
exists a proof net that respects sisterhood.
First we establish the following technical lemma.

\begin{lemma}\label{Lm:badnegative}
For an expression of the form ${} \metapar A_1^- \metapar \ldots \metapar A_m^-$ (which is not a translation of an $\Lb$
sequent, since there is no $B^+$ at the end) there couldn't exist a proof net.
\end{lemma}

\begin{proof}
Following Pentus~\cite{Pentus1998}, for any string $\gamma$ of literals, $\Par$'s, $\Tensor$'s, and $\metapar$'s we
define $\natural(\gamma)$ as the number of negative literals (i.e., of the form $\bar{q}$, where $q$ is a variable or a bracket) minus
the number of $\Par$'s and $\metapar$'s. Then we establish the following:
(1) $\natural(A^+) = 0$ and $\natural(A^-) = 1$ for any formula $A$;
(2) if there exists a proof net for $\gamma$, then $\natural(\gamma) = -1$.
The first statement is proved by joint induction on $A$. The second one follows from the fact that
the number of regions is greater than the number of $\Ec$ links exactly by one; links are in one-to-one correspondence
with negative literal occurrences and each region holds a unique occurrence of $\Par$ or $\metapar$.
Since $\natural({} \metapar A_1^- \metapar \ldots \metapar A_m^-) = m - m = 0 \ne -1$, there is no proof net.
\end{proof}

\begin{theorem}\label{Th:proofnets} 
The sequent $\Gamma \to C$ is derivable in $\Lb$ if and only if there exists a proof net $\Ec$ over $\Omega_{\Gamma \to C}$ that
respects sisterhood.
\end{theorem}

\begin{proof}
The direction from $\Lb$-derivation to proof net is routine:  we construct the proof net by induction,
maintaining the correctness criterion. 
For the other direction, we proceed by induction on the number of $\Par$ and $\Tensor$ occurrences. 

If there are no occurrences of $\Par$ or $\Tensor$,
then the total number of $\metapar$ occurrences is, on the one hand, equal to $n$; on the other hand,
it is equal to the number of regions, $\frac{n}{2} + 1$. Therefore, $n = 2$, and 
the only possible proof net is 
$\xymatrix @-9mm{
\metapar & \bar{p} \ar@{-} '+/u4mm/'[rr]+/u4mm/ [rr] & \metapar & p
}$ that corresponds to the $p \to p$ axiom. 

Otherwise consider the set of all occurrences of
$\Par$ and $\Tensor$ with the relation $\Ac \cup {\prec}$. Since this relation is acyclic (and the
set is not empty), there
exists a maximal element, $c_i$. 

{\em Case 1.1:} $c_i$ is a $\Par$ occurrence that came from $(A \cdot B)^- = B^- \Par A^-$. Replacing this $\Par$
by $\metapar$ corresponds to applying $(\cdot\to)$.

{\em Case 1.2:} $c_i$ is a $\Par$ occurrence that came from $(A \BS B)^+ = A^- \Par B^+$.
Replacing this $\Par$ by $\metapar$ changes ${} \metapar \Gamma^- \metapar A^- \Par B^+$ to
${} \metapar \Gamma^- \metapar A^- \metapar B^+$, which corresponds to applying $(\to\BS)$.
(In the negative translation, formulae in the left-hand side appear in the inverse order.)

{\em Case 1.3:} $c_i$ is a $\Par$ occurrence that came from $(B \SL A)^+ = B^+ \Par A^-$.
Again, replace $\Par$ with $\metapar$ and cyclically transform the net, yielding ${} \metapar A^- \metapar \Gamma^- \metapar B^+$. Then apply $(\to\SL)$.

{\em Case 2.1:} $c_i$ is a $\Tensor$ occurrence that came from $(A \cdot B)^+ = A^+ \Tensor B^+$.
Then $\Ac(c_i)$ is a $\metapar$ occurrence, and the $\Ac$ link splits the proof net
for $\Gamma \to A \cdot B$ 
into two separate proof nets for $\Gamma_1 \to A$ and $\Gamma_2 \to B$ ($\Gamma_1$ and/or $\Gamma_2$ could
be empty, then two or three $\metapar$'s shrink into one):
$$
\xymatrix @-7mm{
\metapar & \Gamma_2^- & \metapar & \Gamma_1^- & \metapar & A^+ & \Tensor 
\ar `u^l[llll]+/u7mm/`l^d[llll] [llll]
& B^+
}
$$
Note that here the fragments before the $\metapar$ occurrence $\Ac(c_i)$ and between $\Ac(c_i)$ and $A$
are negative translations of whole metaformulae ($\Gamma_1$ and $\Gamma_2$), not just substrings with
possibly disbalanced brackets. Indeed, suppose that a pair of sister brackets, $\bar{[}$ and $\bar{]}$, is split between these two fragments.
Then, since $\Ec$ links cannot intersect $\Ac$, the corresponding pair of $[$ and $]$, connected to the original pair by $\Ec$, will also be
split and therefore belong to translations of different formulae. However, they also form a sister pair (our proof net respects sisterhood),
and therefore should belong to one formula (by definition of the translation). Contradiction.

Since for the sequents $\Gamma_1 \to A$ and $\Gamma_2 \to B$ the induction parameter is smaller,
they are derivable, and therefore $\Gamma_1, \Gamma_2 \to A \cdot B$ is derivable
by application of the $(\to\cdot)$ rule.

{\em Case 2.2:} $c_i$ is a $\Tensor$ occurrence from $(A \BS B)^- = B^- \Tensor A^+$. 
Again, the proof net gets split:
\begin{center}
\begin{tabular}{cccc}
$\metapar$ & $\xymatrix @-7mm{
\ldots & \metapar & B^- & \Tensor 
\ar `u_r[rrrr]+/u7mm/`r_d[rrrr] [rrrr]
& A^+ & \metapar & \Pi^- & \metapar & \ldots
}$ 
& $\metapar$ & $C^+$\\
& \upbracefill \\
& $\Delta^-$
\end{tabular}
\end{center}
(As in the previous case, no pair of sister brackets could be split by $\Ac$ here, and $\Pi^-$ is a translation of a whole metaformula.)
The outer fragment provides a proof net for $\Delta \langle B \rangle \to C$; applying a cyclic permutation
 to the inner fragment yields a proof net for $\Pi \to A$. The goal sequent, $\Delta \langle \Pi, A \BS B \rangle \to C$,
is obtained by applying the $(\SL\to)$ rule.

The other situation,
\begin{center}
\begin{tabular}{cccc}
$\metapar$ & $\xymatrix @-7mm{
\ldots & \metapar & \Pi^- & \metapar & B^- & \Tensor 
\ar `u^l[llll]+/u7mm/`l^d[llll] [llll]
& A^+  & \metapar & \ldots
}$ 
& $\metapar$ & $C^+$\\
& \upbracefill \\
& $\Delta^-$
\end{tabular}
\end{center}
is impossible by Lemma~\ref{Lm:badnegative}, applied to the inner net.

{\em Case 2.3:}  $c_i$ is a $\Tensor$ occurrence that came from $(B \SL A)^- = A^+ \Tensor B^-$. Symmetric. 

{\em Case 3.1:} $c_i$ is a $\Par$ occurrence that came from $(\NMod  A)^+ = \bar{]} \Par A^+ \Par \bar{[}$. Then 
we replace two $\Par$'s by $\metapar$'s and cyclically relocate the rightmost $\bar{[}$ with its $\Ec$ link, obtaining
${}\metapar \bar{[} \metapar \Gamma^- \metapar \bar{]} \metapar A^+$ from
${}\metapar \Gamma^- \metapar \bar{]} \Par A^+ \Par \bar{[}$.
This corresponds to an application of $(\to\NMod )$. 

{\em Case 3.2:} $c_i$ is a $\Par$ occurrence that came from $(\PMod A)^- = \bar{[} \Par A^- \Par \bar{]}$. 
By replacing $\Par$'s with $\metapar$'s, we change $\PMod A$ into $[A]$. This corresponds to an application of $(\PMod\to)$.

{\em Case 4.1:} $c_i$ is a $\Tensor$ occurrence that came from $(\PMod A)^+ = {]} \Tensor A^+ \Tensor {[}$. Consider the $\Ec$ 
links that go from these $]$ and $[$. Since $(\PMod A)^+$ is the rightmost formula, they both either go to the left or 
into $A^+$. The second situation is impossible, because then $\Ac(c_i)$ should also be a $\Par$ occurrence in $A^+$, that
violates the maximality of $c_i$ (and also the acyclicity condition).

In the first situation, the picture is as follows:
$$
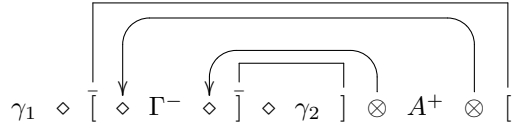
\xymatrix @-8mm{
& \gamma_1 & \metapar & \bar{[} & \metapar & \Gamma^- & \metapar & \bar{]} & \metapar & \gamma_2 & 
] \ar@{-} '+/u6mm/'[lll]+/u6mm/ [lll]
& \Tensor \ar `u^l[lllll]+/u7.7mm/`l^d[lllll] [lllll] 
& A^+ 
& \Tensor \ar `u^l[lllllllll]+/u12mm/`l^d[lllllllll] [lllllllll] & 
[ \ar@{-} '+/u14mm/'[lllllllllll]+/u14mm/ [lllllllllll]
}
$$
(Due to maximality of $c_i$, $\Ac(c_i)$ and its sister are $\metapar$'s, not $\Par$'s.) Clearly,
$\gamma_1$ and $\gamma_2$ are empty: otherwise we have two $\metapar$'s in one region. Then we can remove brackets
and transform this proof net into a proof net for $\metapar \Gamma^- \metapar A^+$, i.e., $\Gamma \to A$. Applying
$(\to\PMod)$  yields $[\Gamma] \to \PMod A$.

{\em Case 4.2:} $c_i$ is a $\Tensor$ occurrence that came from $(\NMod  A)^- = {[} \Tensor A^- \Tensor {]}$. As in the previous case, consider
the $\Ec$ links going from these bracket occurrences. The good situation is when they go to different sides:
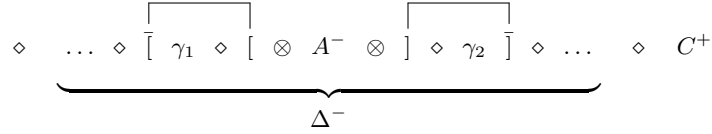
\begin{center}
\begin{tabular}{cccc}
$\metapar$ &
$\xymatrix @-8mm{
\ldots & \metapar & \bar{[} & \gamma_1 & \metapar & 
[ \ar@{-} '+/u6mm/'[lll]+/u6mm/ [lll]
& \Tensor & A^- & \Tensor & 
] \ar@{-} '+/u6mm/'[rrr]+/u6mm/ [rrr]
& \metapar & \gamma_2 & \bar{]} & \metapar & \ldots
}
$ & 
$\metapar$ & $C^+$\\
& \upbracefill \\
& $\Delta^-$
\end{tabular}
\end{center}
(The connectives surrounding ${[} \Tensor A^- \Tensor {]}$ are $\metapar$'s due to the maximality of $c_i$.)
Again, $\gamma_1$ and $\gamma_2$ should be empty 
(the connective after $\bar{[}$ or before $\bar{]}$ here cannot be a $\Tensor$, and we get more than one $\Par$ or $\metapar$ in a region),
and removing the bracket corresponds to applying $(\NMod \to)$: replace $[ \NMod  A ]$ with $A$ in the context $\Delta$.

Potentially, the $\Ec$ links from the brackets could also go to one side, but then they end at a pair of sister brackets $\bar{]}$ and
$\bar{[}$ (in this order), which can occur only in $\bar{]} \Par B \Par \bar{[}$. Then one of these $\Par$'s is $\Ac(c_i)$, which 
contradicts the maximality of $c_i$.
\end{proof}

The idea of proof nets as a representation of derivation in a parallel way comes from Girard's original paper on
linear logic~\cite{Girard1987}. For the non-commutative case, including the Lambek calculus, proof nets were
studied by many researchers including Abrusci~\cite{Abrusci1995}, de~Groote~\cite{philippe:CADE-1999}, Nagayama and Okada~\cite{Nagayama2003}, 
Penn~\cite{Penn2002}, Pentus~\cite{Pentus1998}, Yetter~\cite{Yetter1990},
and others.
In our definition of proof nets for $\Lb$ we follow
Fadda and Morrill~\cite{FaddaMorrill2005}, but with the correctness (acyclicity) conditions of Pentus~\cite{Pentus1998}\cite{Pentus2010} 
rather than Danos and Regnier~\cite{DanosRegnier1989}.

The idea of handling brackets similarly to variables is due to Versmissen~\cite{Versmissen1996}. If we take a sequent that is derivable in
$\Lb$, replace brackets with fresh variables, say, $r$ and $s$, and respectively substitute $r \cdot A \cdot s$ for $\PMod A$ and
$r \BS A \SL s$ for $\NMod A$, we obtain a sequent that is derivable in $\LL$  (this follows from our proof net criterion and
can also be shown directly). Versmissen, however, claims that the converse is also true. This would make 
our Theorem~\ref{Th:main} a trivial corollary of Pentus' result~\cite{Pentus2010}, but Fadda and Morrill~\cite{FaddaMorrill2005} present
a counter-example to Versmissen's claim. Namely, the sequent $[ \NMod p ], [ \NMod q ] \to \PMod \NMod (p \cdot q)$ is not derivable in $\Lb$,
but its translation, $r, r \BS p \SL s, s, r, r \BS q \SL s, s \to r \cdot (r \BS (p \cdot q) \SL s) \cdot s$, is derivable in $\LL$.
This example shows the importance of the sisterhood condition in Theorem~\ref{Th:proofnets}: the only possible proof net for this sequent,
shown below, doesn't respect sisterhood.
$$
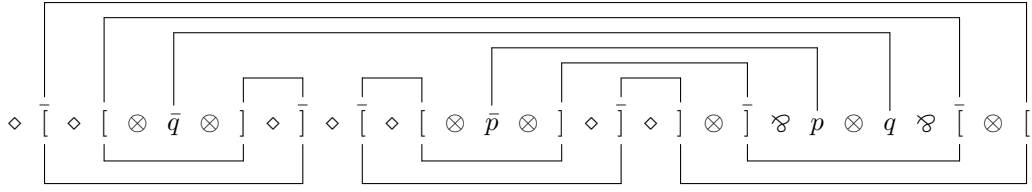
\xymatrix @-8mm{
\metapar & \bar{[} \ar@{-} '+/u16mm/'[rrrrrrrrrrrrrrrrrrrrrrrrrrrrrr]+/u16mm/ [rrrrrrrrrrrrrrrrrrrrrrrrrrrrrr]
\ar@{-} '+/d8mm/'[rrrrrrrr]+/d8mm/ [rrrrrrrr]
 & \metapar & {[} \ar@{-} '+/u14mm/'[rrrrrrrrrrrrrrrrrrrrrrrrrr]+/u14mm/ [rrrrrrrrrrrrrrrrrrrrrrrrrr]
 \ar@{-} '+/d5mm/'[rrrr]+/d5mm/ [rrrr]
 & \Tensor & \bar{q} \ar@{-} '+/u12mm/'[rrrrrrrrrrrrrrrrrrrrrr]+/u12mm/ [rrrrrrrrrrrrrrrrrrrrrr]
 & \Tensor & {]}  \ar@{-} '+/u6mm/'[rr]+/u6mm/ [rr]
 & \metapar & \bar{]}  &
\metapar & \bar{[}  \ar@{-} '+/u6mm/'[rr]+/u6mm/ [rr] \ar@{-} '+/d8mm/'[rrrrrrrr]+/d8mm/ [rrrrrrrr] & \metapar & {[}
\ar@{-} '+/d5mm/'[rrrr]+/d5mm/ [rrrr] & \Tensor & \bar{p} \ar@{-} '+/u10mm/'[rrrrrrrrrr]+/u10mm/ [rrrrrrrrrr] 
& \Tensor & {]} \ar@{-} '+/u8mm/'[rrrrrr]+/u8mm/ [rrrrrr] & \metapar & \bar{]} \ar@{-} '+/u6mm/'[rr]+/u6mm/ [rr] &
\metapar & {]}  \ar@{-} '+/d8mm/'[rrrrrrrrrr]+/d8mm/ [rrrrrrrrrr] & \Tensor & \bar{]} 
\ar@{-} '+/d5mm/'[rrrrrr]+/d5mm/ [rrrrrr]& \Par & p & \Tensor & q & \Par & \bar{[} & \Tensor & {[}
}
$$

\section{Complexity Parameters for Proof Nets}\label{S:params}

In this section we introduce new complexity parameters that operate with $\Omega_{\Gamma \to C}$ rather than
with the original sequent $\Gamma \to C$, and therefore are more handy for complexity estimations. We show that
a value is polynomial in terms of the old parameters if it is polynomial in terms of the new ones.  

The first parameter, denoted by $n$, is the number of literals in $\Omega_{\Gamma \to C}$. It is connected to
the size of the original sequent by the following inequation:
$
n \le 2\nrm{\Gamma \to C}.
$
(We have to multiply by 2, since a modality, $\PMod$ or $\NMod$, being counted as one symbol in $\Gamma \to C$,
introduces two literals.)

The second parameter, denoted by $d = d(\Omega_{\Gamma \to C})$, is the {\em connective alternation depth,} and
informally it is the maximal number of alternations between $\Par$ and $\Tensor$ on the $\prec$-path from any literal
to the root of the parse tree. Formally it is defined by recursion.

\begin{definition}
For an expression $\gamma$ constructed from literals using $\Par$, $\Tensor$, and $\metapar$, let
$\prd(\gamma)$ be  1 if $\gamma$ is of the form $\gamma_1 \Tensor \gamma_2$, and 0 otherwise.
Define $d(\gamma)$ by recursion:
$d(q) = d(\bar{q}) = 0$ for any literal $q$;
$d(\gamma_1 \Par \gamma_2) = d(\gamma_1 \metapar \gamma_2) = \max \{ d(\gamma_1) + \prd(\gamma_1), d(\gamma_2) + \prd(\gamma_2) \}$;
$d(\gamma_1 \Tensor \gamma_2) = \max \{ d(\gamma_1), d(\gamma_2) \}$.
\end{definition}

\noindent
The $d$ parameter is connected to the order of the original sequent: 

\begin{lemma}
For any sequent $\Gamma \to C$, the following  holds:
$
d(\Omega_{\Gamma \to C}) \le \ord(\Gamma \to C).
$
\end{lemma}

\begin{proof}
We prove the following two inequations for any formula $A$ by simultaneous induction on the construction of $A$:
$d(A^+) \le \ord(A)$ and $d(A^-) + \prd(A^-) \le \ord(A)$.
The induction is straightforward, the only interesting cases are the negative ones for  $A_2 \BS A_1$,
$A_1 \SL A_2$, and  
$(\NMod A_1)^-$. In those three cases, we need to branch further into two subcases: whether $A_1$ is a variable or a complex
type.
Finally,
$d(\Omega_{\Gamma \to C}) = d({}\metapar A_1^- \metapar \ldots \metapar A_k^- \metapar C^+) =
\max \{ d(A_1^-) + \prd(A_1^-), \ldots, d(A_k^-) + \prd(A_k^-), d(C^+) + \prd(C^+) \} \le 
\max \{ \ord(\Gamma), \ord(C) + \prd(C) \} \le 
\max \{ \ord(\Gamma) + 1, \ord(C) + \prd(C) \} = \ord(\Gamma \to C)$.
\end{proof}

The third parameter is $b$, the maximal nesting depth of pairs of sister brackets. Clearly,
$b = \br(\Gamma \to C)$.

In view of the inequations established in this section, if a value is $\poly(n, 2^d, n^b)$, it is
 also $\poly(\nrm{\Gamma \to C}, 2^{\ord(\Gamma \to C)}, \nrm{\Gamma \to C}^{\br(\Gamma \to C)})$, and for
 our algorithm we'll establish complexity bounds in terms of $n$, $d$, and $b$.

\section{The Algorithm}\label{S:algo}
Our goal is to obtain an efficient algorithm that searches for proof nets that respect sisterhood.
We are going to split this task: first find all possible proof nets satisfying Pentus' correctness
conditions, and then distill out those which respect sisterhood. One cannot, however, simply yield all
the proof nets. The reason is that there exist derivable sequents, even without brackets and of order 2, that
have exponentially many proof nets, for example, $p \SL p, \dots, p \SL p, p, p \BS p, \dots, p \BS p \to p$.
Therefore, instead of generating all the proof nets for a given sequent, Pentus, as a side-effect of his 
provability verification algorithm, produces a context free grammar that generates a set of words encoding all these proof nets. 
We filter this set by intersecting it
with the set of codes of all proof structures that respect sisterhood. For the latter, we
build a finite automaton of polynomial size.

Note that this context free grammar construction is {\em different} from the translation of Lambek
categorial grammars into context free grammars (Pentus~\cite{pentus:lccf}). The grammar from Pentus' algorithm that 
we consider here generates all proof nets for a {\em fixed} sequent, while in~\cite{pentus:lccf} a context free grammar
is generated for {\em all} words that have corresponding derivable Lambek sequents. The latter (global) grammar is of exponential size
(though for the case of only one division there also exists a polynomial construction~\cite{Kuznetsov2016}), while
the former (local) one is polynomial. For the bracket extension, we present a construction of the local grammar. The
context-freeness for the global case is claimed by J\"ager~\cite{Jaeger2003}, but his proof uses the incorrect
lemma by Versmissen (see above). 

Following Pentus~\cite{Pentus2010}, for a given sequent $\Gamma \to C$ we encode proof structures as words of length $n$ over alphabet
$\{ e_1, \dots, e_n \}$. 

\begin{definition}
The code $c(\Ec)$ of proof structure $\Ec$ is constructed as follows: if $\ell_i$ and $\ell_j$ 
are connected by $\Ec$, then the $i$-th letter of $c(\Ec)$ is $e_j$ and the $j$-th letter is $e_i$. 
\end{definition}
The code of a proof structure is always an involutive permutation of $e_1$, \dots, $e_n$.

We are going to define two
languages, $P_1$ and $P_2$, with the following properties:
\begin{enumerate}
\item $P_1 = \{ c(\Ec) \mid \mbox{$\Ec$ is a proof net}\}$;
\item $c(\Ec) \in P_2 \mbox{\quad if{f} \quad} \mbox{$\Ec$ respects sisterhood}$.
\end{enumerate}

Note that in the condition for $P_2$ we say nothing about words that are not of the form $c(\Ec)$. Some of these
words could also belong to $P_2$. Nevertheless, $w \in P_1 \cap P_2$ if{f} $w = c(\Ec)$ 
for some pairing $\Ec$ that is a proof net and respects sisterhood. Therefore, the sequent is derivable in $\Lb$ if{f}
$P_1 \cap P_2 \ne \varnothing$.

Now the algorithm that checks derivability in $\Lb$ works as follows: it constructs a context free grammar for
$P_1 \cap P_2$ and checks whether the language generated by this grammar is non-empty. Notice that the existence of
such a grammar is trivial, since the language is finite. However, it could be of exponential size, 
and we're going to construct a grammar of size $\poly(n, 2^d, n^b)$, and do it in polynomial time.

For $P_1$, we use the construction from~\cite{Pentus2010}. As the complexity measure (size) of a context free grammar, $|\Gf_1|$,
we use the sum of the length of its rules.

\begin{theorem}[M.~Pentus 2010]\label{Th:PentusGrammar}
There exists a context free grammar $\Gf_1$ of size $\poly(n, 2^d)$ that generates $P_1$. Moreover,
this grammar can be obtained from the original sequent by an algorithm with working time
also bounded by $\poly(n, 2^d)$.
\end{theorem}

This theorem is stated as a remark in~\cite{Pentus2010}. We give a full proof of it in Section~\ref{S:Pentus}.

Next, we construct a finite automaton for a language that satisfies the condition for $P_2$.

\begin{lemma}\label{Lm:automaton}
There exists a deterministic finite automaton with $\poly(n,  n^b)$ states that generates a language $P_2$ over
alphabet $\{ e_1, \dots, e_n \}$ such that $c(\Ec) \in P_2$ if{f} $\Ec$ respects sisterhood.
Moreover, this finite automaton can be obtained from the original sequent by an algorithm with
working time $\poly(n, n^b)$.
\end{lemma}

\begin{proof}
First we describe this automaton informally. Its memory is organised as follows: it includes a pointer $i$ to the current letter
of the word (a number from 1 to $n+1$) and a stack that can be filled with letters of $\{ e_1, \dots, e_n \}$. In the beginning,
$i = 1$ and the stack is empty. At each step (while $i \leq n$), the automaton looks at $\ell_i$. 
If it is not a bracket, the automaton increases the pointer and proceeds to the next letter in the word. If it is a bracket, let 
its sister bracket be $\ell_j$.
Denote the $i$-th (currently being read) letter of the word by $e_{i'}$. If $\ell_{i'}$ is not a bracket,
yield ``no'' (bracket is connected to non-bracket). Otherwise let $\ell_{j'}$ be the sister of $\ell_{i'}$
and consider two cases.
\begin{enumerate}
\item $j>i$. Then push $e_{j'}$ on top of the stack, increase the pointer and continue.
\item $j<i$. Then pop the letter from the top of the stack and compare it with $e_{i'}$. If they do not coincide,
yield ``no''. Otherwise increase the pointer and continue.
\end{enumerate}
If $i = n+1$ and the stack is empty, yield ``yes''.

Since sister brackets are well-nested, on the $i$-th step we pop from the stack the symbol that was pushed there on the $j$-th step
(if $\ell_i$ and $\ell_j$ are sister brackets and $j<i$). Thus, the symbol popped from the stack contains exactly the information
that, if the bracket $\ell_j$ is connected to $\ell_{j'}$, then the bracket $\ell_i$ should be connected to the sister bracket
$\ell_{i'}$, and we verify the fact that $\Ec$ satisfies this condition by checking that the $i$-th letter is actually $e_{i'}$.

Note that here we do not check the fact that the word really encodes some proof structure $\Ec$, since malformed codes will be ruled out by
the intersection with $P_1$.

If the bracket nesting depth is $b$, we'll never have more than $b$ symbols on the stack. For each symbol we have $n$ possibilities
($e_1$, \dots, $e_n$). Therefore, the total number of possible states of the stack is $1 + n + n^2 + \ldots + n^b \leq (b+1) \cdot n^b$.
The pointer has $(n+1)$ possible values. Thus, the whole number of possible memory states is $(n+1) \cdot (b+1) \cdot n^b + 1$
(the last ``$+1$'' is for the ``failure'' state, in which the automaton stops to yield ``no'').

Formally, our automaton is a tuple $\Af_2 = \langle Q, \Sigma, \delta, q_0, \{ q_F \} \rangle$, where $\Sigma = \{ e_1, \dots, e_n \}$ is the alphabet,
$Q = \{ 1, \dots, n+1 \} \times \Sigma^{\leq b} \cup \{ \bot \}$, where $\Sigma^{\leq b}$ is the set of all words over $\Sigma$ of length
not greater than $b$, is the set of possible states  ($\bot$ is the ``failure'' state), $q_0 = \langle 1, \varepsilon \rangle$ is the
initial state, $q_F = \langle n+1, \varepsilon \rangle$ is the final (accepting) state, and $\delta \subset Q \times \Sigma \times Q$ is a set
of transitions defined as follows:
\begin{align*}
\delta &= \{ \langle i, \xi \rangle \xrightarrow{e_{i'}} \langle i+1, \xi \rangle  \mid
\mbox{$\ell_i$ is not a bracket} \} \\
&\cup \{ \langle i, \xi \rangle \xrightarrow{e_{i'}} \bot \mid 
\mbox{$\ell_i$ is a bracket and $\ell_{i'}$ is not a bracket} \}\\
&\cup \{ \langle i, \xi \rangle \xrightarrow{e_{i'}} \langle i+1, \xi e_{j'} \rangle \mid
\mbox{$\ell_i$ is a bracket, its sister bracket is $\ell_j$, $j>i$;}\\&\qquad\qquad\mbox{ 
$\ell_{i'}$ is a bracket, its sister bracket is $\ell_{j'}$}\}\\
&\cup \{ \langle i, \xi e_{i'} \rangle \xrightarrow{e_{i'}} \langle i+1, \xi \rangle
\mid \mbox{$\ell_i$ is a bracket, its sister bracket is $\ell_j$, $j<i$} \}\\
&\cup \{ \langle i, \xi e_{i''} \rangle \xrightarrow{e_{i'}} \bot \mid
\mbox{$\ell_i$ is a bracket, its sister bracket is $\ell_j$, $j<i$, and $i' \ne i''$} \}.
\end{align*}
$\Af_2$ is a deterministic finite automaton with not more than $(n+1) \cdot (b+1) \cdot n^b + 1 =
\poly(n, n^b)$ states, and it generates a language $P_2$ such that $c(\Ec) \in P_2$ if{f}
$\Ec$ respects sisterhood.

In the RAM model, each transition is computed in constant time, and the total number of transitions is not more than
$|Q|^2 \cdot n \le ((n+1) \cdot (b+1) \cdot n^b + 1)^2 \cdot n$, which is also $\poly(n, n^b)$.
\end{proof}

Now we combine Theorem~\ref{Th:PentusGrammar} and Lemma~\ref{Lm:automaton} to obtain a context free grammar $\Gf$ for $P_1 \cap P_2$ of 
size $\poly(|\Gf_1|, |\Af_2|, |\Sigma|)$, where $|\Af_2|$ is  the number of states of $\Af_2$. For this we use the following well-known result:
\begin{theorem}\label{Th:CFGintersection}
If a context free grammar $\Gf_1$ defines a language $P_1$ over an alphabet $\Sigma$ and
a deterministic finite automaton $\Af_2$  defines a language $P_2$ over the same alphabet, then there
exists a context free grammar $\Gf$ that defines $P_1 \cap P_2$, the size of this grammar
is $\poly(|\Gf_1|, |\Af_2|, |\Sigma|)$, and, finally, this grammar can be obtained from
$\Gf_1$ and $\Af_2$ by an algorithm with working time also $\poly(|\Gf_1|, |\Af_2|, |\Sigma|)$.
\end{theorem}

\noindent
For this theorem we use the construction 
from~\cite[Theorem~3.2.1]{Ginsburg1966} that  works directly with the context free formalism. 
This makes the complexity estimation straightforward.
Since $|\Gf_1| = \poly(n, 2^d)$, $|\Af_2| = \poly(n, n^b)$,
and $|\Sigma| = n$, $|\Gf_2|$ is $\poly(n, 2^d, n^b)$.
Finally, checking derivability of the sequent is equivalent to testing the language $P_1 \cap P_2$ for
non-emptiness, which is done
using the following theorem~\cite[Lemma~1.4.3a and Theorem~4.1.2a]{Ginsburg1966}: 
\begin{theorem}\label{Th:CFGnonemptiness}
There exists an algorithm that checks whether the language generated by a context free grammar $\Gf$ is non-empty,
with $\poly(|\Gf|)$ working time.
\end{theorem}

\noindent The whole algorithm described in this section works in $\poly(n, 2^d, n^b) = 
\poly(\nrm{\Gamma \to C},\linebreak 2^{\ord(\Gamma \to C)}, \nrm{\Gamma \to C}^{\br(\Gamma \to C)})$ time, as required in Theorem~\ref{Th:main}.

\section{Proof of Theorem~\ref{Th:PentusGrammar} (Pentus' Construction Revisited)}\label{S:Pentus}

In our algorithm, described in Section~\ref{S:algo}, we use Pentus' polynomial-size context free grammar,
that generates all proof nets, as a black box: we need only Theorem~\ref{Th:PentusGrammar} itself, not the
details of the construction in its proof. However, Pentus~\cite{Pentus2010} doesn't explicitly formulate this
theorem, but rather gives it as side-effect of the construction for checking {\em existence} of a proof net
(e.g., non-emptiness of the context free language). The latter is, unfortunately, not sufficient for our needs.
Moreover, we use slightly different complexity parameters.
Therefore, and also in order to make our paper logically self-contained, in this section we redisplay Pentus' construction in more detail.
In other words, here we prove Theorem~\ref{Th:PentusGrammar}.

Pentus' idea for seeking proof nets is based on dynamic programming. In $\Omega_{\Gamma \to C}$, connective and literal
occurrences alternate: $c_1, \ell_1, c_2, \ell_2, \ldots, c_n, \ell_n$. Consider triples of the form $(i,j,k)$, where
$1 \le i \le j \le k \le n$.

\begin{definition}
An {\em $(i,j,k)$-segment} $\tilde{\Ec}$ is a planar pairing of literals from $\{ \ell_i, \ldots, \ell_{k-1} \}$
such that in every region created by $\tilde{\Ec}$ there exists a unique $\Par$ or $\metapar$ occurrence from
$\{ c_i, \ldots, c_k \}$ that belongs to this region, and, in particular, this occurrence for the outer (infinite)
region is $c_j$, called the {\em open par.} 
If $k = j = i$, then $\tilde{\Ec}$ is empty, and $c_i$ should be a $\Par$ or $\metapar$ occurrence.
\end{definition}

For each $\tilde{\Ec}$ we construct the corresponding $\tilde{\Ac}$ that connects each $\Tensor$ occurrence to the
only $\Par$ or $\metapar$ in the same region; for the outer region, it uses the open par.

\begin{definition}
An $(i,j,k)$-segment $\tilde{\Ec}$ is {\em correct,} if{f} the graph $\tilde{\Ac} \cup {\prec}$ is acyclic.
\end{definition}

For each $(i,j,k)$-segment, in the non-terminals of the grammar we keep a small amount of information,
which we call the {\em profile} of the segment and that is sufficient to construct bigger segments (and, finally,
the whole proof net) from smaller ones.

\begin{definition}
A $\Tensor$ occurrence is called {\em dominant,} if{f} it is not immediately dominated (in the $\prec$ preorder) by another
$\Tensor$ occurrence. For each $\Tensor$ occurrence $c$ there exists a unique dominant $\Tensor$ occurrence $\tau(c)$ such
that $\tau(c) \succeq c$ and on the $\prec$ path from $c$ to $\tau(c)$ all occurrences are $\Tensor$ occurrences.
\end{definition}

Let's call two $\Tensor$ occurrences {\em equivalent,} $c \approx c'$, if $\tau(c) = \tau(c')$.
Equivalent $\Tensor$ occurrences form {\em clusters;} from each cluster
we pick a unique representative, the $\prec$-maximal occurrence $\tau(c)$. By $\precsim$ we denote the transitive closure
of ${\prec} \cup {\approx}$: $c \precsim c'$ means that there is a path from $c$ to $c'$ that goes 
along $\prec$ and also could go in the inverse direction, but only from $\Tensor$ to $\Tensor$ with no $\Par$ or $\metapar$ in
between.

\begin{lemma}
For an $(i,j,k)$-segment $\tilde{\Ec}$, the graph $\tilde{\Ac} \cup {\prec}$ is acyclic if{f}
any cycle in $\tilde{\Ac} \cup {\precsim}$ is a trivial $\approx$-cycle in a cluster,
and similarly for a proof structre $\Ec$ and the graphs $\Ac \cup {\prec}$ and
$\Ac \cup {\precsim}$.
\end{lemma}

\begin{proof}
Pentus proves this lemma by a topological argument. If $\tilde{\Ac} \cup {\precsim}$ has non-trivial cycles, 
take a simple cycle (i.e. a cycle where no vertex appears twice) that embraces the smallest area. If this cycle
includes a link from $c$ to $d$ where $c \approx d$ and $c \succ d$, then consider the $\tilde{\Ac}$ link that goes
from $c$. This link should go {\em inside} the cycle, and, continuing by this link, one could construct a new cycle with
a smaller area embraced. Contradiction. The other direction is trivial, since every cycle in $\tilde{\Ac} \cup {\prec}$ is
a non-trivial cycle in $\tilde{\Ac} \cup {\precsim}$.
\end{proof}

\noindent
In view of this lemma we can now use $\precsim$ instead of $\prec$ in the correctness (acyclicity) criteria for
proof nets and $(i,j,k)$-segments.

\begin{definition}
For a connective occurrence $c_i$ let $V_i$ be the set of all dominant $\Tensor$ occurrences on the $\prec$ path from 
$c_i$ to the root of the parse tree.
\end{definition}

\noindent
Since each dominant $\Tensor$ marks a point of alternation between $\Tensor$ and $\Par$ (or $\metapar$, on the top level),
and the number of such alternations is bounded by $d$, we have $|V_i| \le d$ for any $i$.

\begin{definition}
The profile of an $(i,j,k)$-segment $\tilde{\Ec}$, denoted by $R$, is the restriction of the transitive closure of $\tilde{\Ac} \cup {\precsim}$ to
the set $V_i \cup V_j \cup V_k$ that is forced to be irreflexive (in other words, we remove trivial $\approx$-cycles). 
An {\em $(i,j,k)$-profile} is an arbitrary transitive irreflexive relation on $V_i \cup V_j \cup V_k$.
\end{definition}

\begin{lemma}\label{Lm:ghost_count}
The number of different $(i,j,k)$-profiles is $\poly(2^d)$.
\end{lemma}

\begin{proof}
Let $|V_i| = d_1$, $|V_j| = d_2$, $|V_k| = d_3$ (these three numbers are not greater than $d$).
Each profile includes three chains, $Q_i$, $Q_j$, and $Q_k$, and it remains to count the number of possible
connections between them. Due to transitivity, if a vertex in $V_i$ is connected to a vertex in $V_j$, then
it is also connected to all greater vertices. Now we represent elements of $V_j$ as $d_2$ white balls, and put 
$d_1$ black balls between them. The $i$-th black ball is located in such a place that the $i$-th vertex of $V_i$
is connected to all vertices of $V_j$ that are greater than the position of the $i$-th ball, and only to them.
Due  to transitivity,
the order of black balls is the same as $Q_i$. The number of possible distributions of white and black balls is
${{d_1+d_2} \choose d_1} < 2^{d_1+d_2} \le 2^{2d}$. Doing the same for all 6 pairs of 3 chains, we get
the estimation $(2^{2d})^6 = (2^d)^{12} = \poly(2^d)$ for the number of $(i,j,k)$-profiles.
\end{proof}

Now we define the context free grammar $\Gf_1$. Non-terminal symbols of this grammar include the starting symbol $S$ and 
symbols $F_{i,j,k,R}$ for any triple $(i,j,k)$ ($1 \le i \le j \le k \le n$) and any $(i,j,k)$-profile $R$. The meaning of these non-terminals is in the following statement, which will be proved by induction
after we present the rules of $\Gf_1$: a word $w$ is derivable from $F_{i,j,k,R}$ if{f} $w = c(\tilde{\Ec})$ for a correct
$(i,j,k)$-segment $\tilde{\Ec}$ with profile $R$; a word $w$ is derivable from $S$ if{f} $w = c(\Ec)$ for some proof net $\Ec$.
(Codes of $(i,j,k)$-segments are defined in the same way as codes of proof structures, as involutive permutations of $e_i$, \dots, $e_{k-1}$.)

For the induction base case, $i = j = k$, we take only those values of $i$ such that $c_i$ is a $\Par$ or $\metapar$ occurrence,
and denote by $Q_i$ the $\prec$ relation restricted to $V_i$  (this is the trivial profile of an empty $(i,i,i)$-segment);
$Q_i$ is always acyclic, and an isolated $\Par$ or $\metapar$ occurrence $c_i$ is always a correct $(i,i,i)$-segment (with an empty
$\tilde{\Ec}$), and $c_i$ is its open par.
Now for each $\Par$ or $\metapar$ occurrence $c_i$ we add the following rule to the grammar (this is a $\varepsilon$-rule, the right-hand side is empty):
$$
F_{i,i,i,Q_i} \Rightarrow.
$$

Next, consider the non-trivial situation, where $i < k$.  The difference $k-i$ should be even, otherwise there
couldn't exist a literal pairing $\tilde{\Ec}$. Moreover, if both $c_i$ and $c_k$ are $\Par$ or $\metapar$ occurrences, 
 a correct $(i,j,k)$-segment couldn't exist either, since in the outer region we have at least two $\Par$ or
$\metapar$ occurrences, namely, $c_i$ and $c_k$. Therefore, we include rules for $F_{i,j,k,R}$ only if 
$k-i$ is even and at least one of $c_i$ and $c_k$ should be $\Tensor$. Let it be $c_i$ (Pentus' {\em situation of
the first kind}). The $c_k$ case (Pentus' {\em situation of the second kind}) is handled symmetrically.

We take the leftmost literal occurrence, $\ell_i$, and find all possible occurrences among $\ell_{i+1}$, \ldots, $\ell_{k-1}$ that
could be connected to $\ell_i$ (i.e., if $\ell_i$ is an occurrence of $q$, we seek $\bar{q}$, and vice versa). For each such
occurrence, $\ell_{h-1}$, we consider two triples, $(i+1,j',h-1)$ and $(h,j,k)$, and all possible $(i+1,j',h-1)$- and
$(h,j,k)$-profiles, $R_1$ and $R_2$, respectively. For each such pair, $R_1$ and $R_2$, we consider the transitive closure
of the following relation:
$
R_1 \cup R_2 \cup Q_i \cup \{ \langle \tau(c_i), d \rangle \mid d \in V_j \}.
$
If it is irreflexive (acyclic), its restriction to $V_i \cup V_j \cup V_k$, denoted by $R$,
 will become a profile of an $(i,j,k)$-segment. For this, we add the following rule to
the grammar:
$$
F_{i,j,k,R} \Rightarrow e_{h-1} \: F_{i+1,j',h-1,R_1} \: e_i \: F_{h,j,k,R_2}.
$$

\begin{lemma}
In this grammar, a word $w$ can be derived from $F_{i,j,k,R}$ if{f} $w = c(\tilde{\Ec})$ for some
$(i,j,k)$-segment $\Ec$ with profile $R$.
\end{lemma}

\begin{proof}
Proceed by induction on $k-i$. The base case ($i=j=k$) was considered above.

For {\bf the ``only if'' part,} let $w$ be derived by a rule for the first kind (the second kind is symmetric).
Then $w = e_{h-1} w_1 e_i w_2$, and by induction hypothesis $w_1$ and $w_2$ encode $(i+1,j',h-1)$- and
$(h,j,k)$-segments with profiles $R_1$ and $R_2$ respectively. The word $w$ encodes an $(i,j,k)$-segment,
and it remains to show that this segment is correct and its profile is $R$. For this new segment,
$\tilde{\Ac} = \tilde{\Ac}_1 \cup \tilde{\Ac}_2 \cup \langle c_i, c_j \rangle$. Suppose there is a non-trivial cycle
in $\tilde{\Ac} \cup {\precsim}$. Since all cycles in $\tilde{\Ac}_1 \cup {\precsim}$ and $\tilde{\Ac}_2 \cup {\precsim}$ are trivial,
this cycle should either include links from both $\tilde{\Ac}_1$ and $\tilde{\Ac}_2$ or use
the new $\langle c_i, c_j \rangle$ connection (or both). The cycle, however, cannot cross
$\tilde{\Ec}$ links, therefore the only way of ``legally crossing the border'' between segments is by going through
$\Tensor$ occurrences that dominate $c_i$, $c_{i+1}$, $c_{h-1}$, $c_h$, or $c_k$. 
We can assume that these
``border crossing points'' are dominant $\Tensor$ occurrences (otherwise we can add a $\approx$-detour to the cycle).
Then the cycle is actually a concatenation of parts of the following three kinds: (1) connecting vertices of $V_{i+1} \cup V_{h-1}$;
(2) connecting vertices of $V_h \cup V_k$; (3) connecting $\tau(c_i)$, via $c_j$, to a vertex $d$ of $V_j$.
In this case, our cycle induces a cycle in $R_1 \cup R_2 \cup Q_i \cup \{ \langle \tau(c_i), d \rangle \mid d \in V_j\}$, which
is impossible by definition.

It remains to show that $R$ is the profile of the newly constructed segment. Indeed, $R$ is a binary relation on $V_i \cup V_j \cup V_k$
and is included in the transitive closure of $\tilde{\Ac} \cup {\precsim}$, therefore $R$ is a subrelation of the profile.
On the other hand, if there is a pair $\langle c,d \rangle$ in the profile, then there is a path from $c$ to $d$ and, as shown above,
it can be split into parts of kinds (1), (2), and (3). Thus,  $\langle c,d \rangle \in R$, and therefore $R$ coincides with the profile.

For {\bf the ``if'' part,} if $c_i$ in an $(i,j,k)$-segment is a $\Tensor$ occurrence, consider the $\tilde{\Ec}$ link from
the literal occurrence $\ell_i$. It splits the segment into two ones. For each of them, by induction hypothesis, we generate their codes from
$F_{i+1,j',h-1,R_1}$ and $F_{h,j,k,R_2}$ respectively, and then apply the rule to generate the code of the original segment.
Situations of the second kind, where $c_k$ is a $\Tensor$ occurrence, are handled symmetrically.
\end{proof}

Finally, we add rules for the starting symbol. These rules are analogous to the rules for situations of the second kind.
Take $\ell_n$ and find all possible occurrences among $\ell_1$, \dots, $\ell_{n-1}$ that
could be connected to it. For each such occurrence $\ell_h$ and any pair of $(0,0,h)$- and $(h+1,j',n)$-profiles, $R_1$ and $R_2$, 
 respectively (in the first segment $j = 0$, since in the whole proof net the
open par should be the leftmost occurrence of $\metapar$), consider the transitive closure of
$R_1 \cup R_2$. If it is irreflexive, then we add the following rule to the grammar:
$$
S \Rightarrow F_{0,0,h,R_1} \: e_n \: F_{h+1,j',n,R_2} \: e_h.
$$

\begin{lemma}
A word $w$ can be derived from $S$ if{f} $w = c(\tilde{\Ec})$ for some proof net $\Ec$.
\end{lemma}

\begin{proof}
Analogous to the previous lemma.
\end{proof}
 
This lemma shows that we've constructed a grammar that generates $P_1$. Now to finish the proof of Theorem~\ref{Th:PentusGrammar}
it remains to establish complexity bounds. 
The number of non-terminal symbols is bounded by $n^3 \cdot K + 1$, where $K$ is the maximal number of $(i,j,k)$-profiles.
Since each rule has length at most 5 (1 non-terminal on the left and 4 symbols on the right), $|\Gf_1|$ is bounded by $5 (n^3 \cdot K + 1)$,
and, since $K$ is $\poly(2^d)$ (Lemma~\ref{Lm:ghost_count}),  $|\Gf_1|$ is $\poly(2^d, n)$. Clearly, the procedure that generates $\Gf_1$ from
the original sequent is also polynomial in running time: acyclicity checks for each rule are performed in $\poly(n)$ time,
and the number of rules is $\poly(2^d,n)$.
 
\section{Conclusions and Future Work}\label{S:future}

In this paper, we've presented an algorithm for provability in the Lambek calculus with brackets.
Our algorithm runs in  polynomial time w.r.t.\ the size of the input sequent, if its order and bracket nesting depth are
bounded. Our new result for bracket modalities is non-trivial, and we address
it with a combination of proof nets and finite automata techniques.

We summarize some questions raised for future research.
First, Pentus~\cite{Pentus2010} also presents a parsing procedure 
for Lambek categorial grammars. In a Lambek grammar, several types can be assigned to one word,
which adds an extra level of non-determinism.
Our intention is to develop an efficient parsing procedure for grammars with brackets. 
Second, the problem whether $\Lb$-grammars define exactly context free languages is still open
(the counter-example by Fadda and Morrill~\cite{FaddaMorrill2005}  jeopardises J\"ager's claim).
Third, in our calculus we allow empty antecedents. We are going to modify our algorithm for the bracketed extension of the original Lambek calculus,
using a modified notion of proof nets (see for example~\cite{LamarcheRetore1996}\cite{Kuznetsov2012}).
A more general question is to extend the algorithm to other enrichments of the Lambek calculus (see, for example,~\cite{Morrill2011}),
keeping polynomiality, if possible. Notice that some of these enrichments are generally undecidable~\cite{KKSarXiv}, so it is
interesting to find feasible bounded fragments. 

\subparagraph*{Acknowledgements.} 
The work of M. Kanovich and A. Scedrov was supported by the Russian Science Foundation under grant 17-11-01294 and performed at National Research University Higher School of Economics, Russia. The work of G. Morrill was supported by an ICREA Academia 2012
and MINECO TIN2014-57226-P (APCOM). The work of S. Kuznetsov was supported by the Russian Foundation for Basic Research
(grant 15-01-09218-a) and by the Presidential Council for Support of Leading Research Schools
(grant N\v{S}-9091.2016.1). Section 1 was contributed by G. Morrill, Section 4 by M. Kanovich and A. Scedrov, Section 5 by S. Kuznetsov, and Sections 2, 3, 6, and 7 were contributed jointly and equally by all coauthors.

The authors are grateful to Mati Pentus for in-depth comments on his algorithm~\cite{Pentus2010}.

\bibliography{Lb-poly}

\end{document}